\documentclass[aps,pra,superscriptaddress,twocolumn,nofootinbib,notitlepage,tightenlines]{revtex4-1}
\usepackage{graphicx}
\usepackage{array, color}
\usepackage{amsthm}
\usepackage{amsmath}
\usepackage{amssymb}
\usepackage{mathrsfs}
\usepackage{txfonts}
\usepackage{pgfplots}
\newtheorem{definition}{Definition}
\newtheorem{theorem}{Theorem}
\newtheorem{proposition}{Proposition}

\newtheorem{corollary}{Corollary}

\newtheorem{remark}{Remark}

\begin{document}

\title{Projective robustness for quantum channels and measurements and their operational significance}

\author{Mingfei Ye}
\email{yemingfei1223@163.com}
\affiliation{School of Mathematics and Statistics, Shaanxi Normal University, Xi'an 710062, China}

\author{Yu Luo}
\email{luoyu@snnu.edu.cn}
\affiliation{College of computer Science, Shaanxi Normal University, Xi'an 710062, China}

\author{Zhihui Li}
\email{lizhihui@snnu.edu.cn}
\affiliation{School of Mathematics and Statistics, Shaanxi Normal University, Xi'an 710062, China}

\author{Yongming Li}
\email{liyongm@snnu.edu.cn}
\affiliation{School of Mathematics and Statistics, Shaanxi Normal University, Xi'an 710062, China}

\date{\today}

\begin{abstract}
\noindent \textbf{Abstract.} Recently, the projective robustness of quantum states has been introduced in [arXiv:2109.04481(2021)]. It shows that the projective robustness is a useful resource monotone and can comprehensively characterize capabilities and limitations of probabilistic protocols manipulating quantum resources deterministically. In this paper, we will extend the projective robustness to any convex resource theories of quantum channels and measurements. First, we introduce the projective robustness of quantum channels and prove that it satisfies some good properties, especially sub- or supermultiplicativity under any free quantum process. Moreover, we use the projective robustness of channels to give lower bounds on the errors and overheads in any channel resource distillation. Meanwhile, we show that the projective robustness of channels quantifies the maximal advantage that a given channel outperforms all free channels in simultaneous discrimination and exclusion of a fixed state ensemble. Second, we define the projective robustness of quantum measurements and prove that it exactly quantifies the maximal advantage that a given measurement provides over all free measurements in simultaneous discrimination and exclusion of two fixed state ensembles. Finally, within a specific channel resource setting based on measurement incompatibility, we show that the projective robustness of quantum channels coincides with the projective robustness of measurement incompatibility.
\end{abstract}
\maketitle

\section{Introduction}

The resource-theoretic framework is a powerful tool to help us understand and manipulate resources. Various specific resource theories have been proposed and discussed under different physical constraints, such as well-known entanglement \cite{Plenio and Virmani 07,Horodecki et.al09}, coherence \cite{Aberg,Baumgratz et.al 14,Streltsov et.al 17}, magic \cite{Veitch et.al 14,Howard and Campbell 17}, non-Gaussianity \cite{Lami18A},  quantum thermodynamics \cite{Brand et.al13,Brand et.al15}, steering \cite{Gallego and Aolita15}, and imaginary \cite{Hickey and Gour 18,Wu et.al 21}. Moreover, due to the generality of this framework, it has been applied to study quantum measurements and quantum channels, and various resource theories about quantum measurements and channels have been proposed, such as measurement simulability and incompatibility \cite{Oszmaniec17,Oszmaniec19,Guerini17}, measurement informativeness \cite{Skrzypczyk19}, measurement coherence \cite{Baek20}, the magic and coherence of quantum channels \cite{Seddon19,Wang19,Saxena20}, and quantum memory of channels \cite{Takagi20}. These resource theories have significantly promoted the prosperous development of quantum informatics. To explore the common characteristics of resources, general resource theories have attracted much attention in recent years, including quantification and manipulation of general resource objects \cite{Chitambar and Gour19,Regula18,FGSL15B,Anshu18,Takagi and Regula 19,Takagi19,Oszmaniec19b,Lami et.al 21,Regula and Lami 21,Regula 21,Liu et.al19,Fang 20,Regula and Bu 20,Regula and Takagi21,Regula and Takagi21b,Ducuara20,Uola20,Ye et.al 21}.

Resource quantification plays a vital role in quantum resource theories. A common and effective approach to accomplish this goal is to define resource measures (or monotones) to measure and compare the resourcefulness contained in a given physical object. At present, ones have proposed plenty of resource measures \cite{Chitambar and Gour19,Regula18}. Such a way can effectively quantify the amount of resource contained in a quantum object, but it often lacks a direct relation with the operational significance of resources \cite{Lami et.al 21}. The generalized robustness measure \cite{Genoni99,Steiner03,Datta09} and weight measure \cite{Lewenstein98} not only fully overcome this shortcoming but also apply to general state resource quantification, even resource quantifications of general measurements and channels. These two measures have wide applications in operational tasks of quantum resource theories, for the former such as catalytic resource erasure \cite{Anshu18}, min-accessible information \cite{Skrzypczyk19,Takagi and Regula 19}, discrimination tasks \cite{Skrzypczyk19,Lami et.al 21,Takagi19,Skrzypczyk19b,Oszmaniec19b,Takagi and Regula 19,Uola19b,Carmeli19b,Mori20,Uola20b,Ducuara20A,Regula and Lami 21}, one shot distillation in general state resource theories \cite{Liu et.al19,Fang 20,Regula and Bu 20}, and one shot manipulation of channel resources \cite{Regula and Takagi21,Regula and Takagi21b}, for the latter such as single-shot information theory \cite{Ducuara20}, exclusion tasks \cite{Ducuara20,Uola20,Ducuara20A,Ye et.al 21}, and channel resource purification \cite{Regula and Takagi21}. In particular, in broad classes of resource theories, they can directly quantify the practical advantages that the resourceful object outperforms all resourceless objects in specific discrimination (or exclusion) tasks \cite{Lami et.al 21,Skrzypczyk19,Takagi and Regula 19,Takagi19,Skrzypczyk19b,Oszmaniec19b,Uola19b,Carmeli19b,Mori20,Uola20b,Regula and Lami 21,Ducuara20,Uola20,Ducuara20A,Ye et.al 21}, and can provide constraints on error rates and overheads for various state and channel distillation protocols \cite{Liu et.al19,Fang 20,Regula and Bu 20,Regula and Takagi21,Regula and Takagi21b}. 
Searching for such resource measures is always one of the cores of resource theories.

Hilbert projective metric introduced by Ref.\cite{Bushell73,Eveson73} has provided a sufficient and necessary condition for probabilistic transformations of pairs of quantum states \cite{Reeb11,Buscemi17}, it has not been applied as a resource monotone until recently projective robustness of quantum states was introduced in \cite{Regula 21}. It is shown that the projective robustness of quantum states not only satisfies a number of useful properties but also has significant applications in any probabilistic manipulation protocol and simultaneous discrimination and exclusion of a fixed channel ensemble. On the one hand, the quantum channels lie at the heart of the manipulation of quantum states. On the other hand, any meaningful quantum information processing task includes a measurement at the end. So it is quite natural and necessary to study the quantification and manipulation of quantum channels and measurements resources. However, the projective robustness has not been studied in the context of channels and measurements resource theories. In this paper, we will introduce two types of projective robustness, namely projective robustness for channels and measurements. For the former, we will study its  operational applications in any deterministic channel distillation protocol and simultaneous discrimination and exclusion of a fixed state ensemble. These two operational tasks form the foundation of the operational aspects of quantum theory. For the latter, we will study its operational applications in simultaneous discrimination and exclusion of two fixed state ensembles. It will promote our understanding of the common characteristics of resources and lay a good foundation for in-depth research to manipulate channel and measurement resources.

This paper is arranged as follows. Section II introduces all of the relevant concepts, including resource theories, some resource monotones, and state discrimination and exclusion tasks. In Sect. III, we define the projective robustness of channels and present that the measure obeys several useful properties. Moreover, we study its two operational applications. In Sect. IV, we introduce the projective robustness of measurements and its corresponding operational interpretation is given.  Particularly, within a specific channel resource theory setting based on measurement incompatibility, we discuss the relationship between projective robustness for channels and projective robustness for measurements. Finally, Section V offers the conclusions.
\section{Setting}
Given the Hilbert space of a finite-dimensional system $A$ with dimension $d_A$, let $\mathbb{L}(A)$ and $\mathbb{D}(A)$ be the set of the linear operators and the set of the density operators acting on this space, respectively. We use $\text{CPTP}(A\rightarrow B)$ to represent the set of quantum channels, i.e. completely positive and trace-preserving (CPTP) maps from $\mathbb{L}(A)$ to $\mathbb{L}(B)$.  For each channel $\mathcal{E}:A\rightarrow B$, its Choi matrix is $J_{\mathcal{E}}:=\mathrm{id} \otimes \mathcal{E}(\Phi^{+}) \in \mathbb{L}(R B)$, where $\Phi^{+}=\sum_{i, j}|i j\rangle\langle i j|$ is the unnormalized maximally entangled state in $\mathbb{L}(R A)$ and $R \cong A$. It is well known that $\mathcal{E}$ is completely positive if and only if $\mathcal{J_E}\geq 0$, trace preserving if and only if $\operatorname{Tr}_R(\mathcal{J_E})=\mathbb{I}$. The normalized Choi state is then $\widetilde{J}_{\mathcal{E}}:=J_{\mathcal{E}} / d_{A}$.  Let $\langle A, B\rangle=\operatorname{Tr}\left(A^{\dagger} B\right)$ represent the Hilbert-Schmidt inner product between operators.  If a map always maps a quantum channel to a valid quantum channel, then it is called a quantum superchannel \cite{Chiribella08}, that is, maps $\operatorname{CPTP}(A \rightarrow B) \rightarrow \operatorname{CPTP}(C \rightarrow D)$.

Throughout the paper, we will be interested in POVMs on finite dimensional Hilbert space $\mathbb{H} \approx \mathbb{C}^{d}$. A POVM $\mathbb{M}$ is a collection of POVM elements $\mathbb{M}=\left\{M_{i}\right\}_{i=1}^n$ with $M_i\geq 0~\forall i$ and $\sum_iM_i=\mathbb{I}$. The operators $M_i$ are called the effects of POVM $\mathbb{M}$. We denote the set of
POVMs on $\mathbb{C}^{d}$ with $n$ outcomes by $\mathcal{M}(d,n)$. This set has a natural convex structure: given two
POVMs $\mathbb{M}, \mathbb{N}\in \mathcal{M}(d,n)$, their convex combination $p\mathbb{M}+(1-p)\mathbb{N}$ is the POVM with $i-$th effect given by $pM_i+(1-p)N_i$. Note that $\mathcal{M}(d,n)$ is a convex subset of $(\operatorname{Herm}(\mathbb{H}))^{\times n}$ where $(\operatorname{Herm}(\mathbb{H}))$ is the set of Hermitian operator from $\mathbb{H}$ to itself. The
space $(\operatorname{Herm}(\mathbb{H}))^{\times n}$ can be regarded as $nd^2$ dimensional real vector space with the inner product defined by $\langle \mathbb{A}, \mathbb{B}\rangle=\sum_i^n \operatorname{Tr}(A_iB_i)$, where $\mathbb{A}=\{A_i\}, \mathbb{B}=\{B_i\}\in (\operatorname{Herm}(\mathbb{H}))^{\times n}$ \cite{Oszmaniec19b}.

\subsection{Resource theories}

A general state resource theory consists of two parts under some physical restrictions \cite{Chitambar and Gour19}. One is the free states which are available freely and we denoted the set of free states as $\mathbb{F}$.  The other is the free operations which are allowed within the given physical constraints. To remain as general as possible, we will only introduce two common and
intuitive assumptions about $\mathbb{F}$: it is convex and closed. The former means that simply taking a sequence of free states generates no resource, and the latter means that no resource can be generated by simply probabilistically mixing free states. As for free operations, we only assume the weakest possible constraint: a free operation does not generate any resource by itself.

When we aim to study channel resources \cite{Gour19a,Gour and Winter19,Liu20,Regula and Takagi21}, in the given physical setting, a specific subset of quantum channels $\mathbb{O}$ was chose as the free channels. We will assume that the set $\mathbb{O}$ is compact and convex. The compactness of $\mathbb{O}$ is a technical assumption that allows us to formally state our results.
Moreover, the free operations between the channels are free superchannels. A free superchannel $\Theta$ does not generate any resource by itself; that is, for any free channel $\mathcal{M}\in \mathbb{O}$, it holds that $\Theta(\mathcal{M})\in \mathbb{O}$. Let $\mathbb{S}$ represent the set of all such free superchannels. The above assumptions will make our results applicable to broad channel resource settings.

As for measurement resource \cite{Takagi and Regula 19,Oszmaniec19b}, we use a similar approach to define the set of free measurements as $\mathcal{M_F}\subseteq \mathcal{M}(d,n)$, and assume that $\mathcal{M_F}$ is convex and compact. A valid transformation $\Pi$ between measurements should transform a POVM into a POVM, i.e., $\Pi(\mathbb{M})=\{\Pi(M_i)\}$ is a valid POVM for any $\mathbb{M}=\{M_i\}\in \mathcal{M}(d,n)$. It means that a valid transformation between measurements should be a positive and trace-preserving map. Moreover, we assume that the class of free operations $\mathcal{O_F}$ consists of mappings $\Gamma : \mathcal{M}(d, n) \rightarrow \mathcal{M}(d, n)$ that (i) preserve the set of free measurements i.e., $\Gamma(\mathbb{N})\in \mathcal{M_F}$ for all $\mathbb{N}\in \mathcal{M_F}$, (ii) are convex-linear i.e., $\Gamma\left(p \mathbb{M}+(1-p) \mathbb{M}^{\prime}\right)=p \Gamma(\mathbb{M})+(1-p) \Gamma\left(\mathbb{M}^{\prime}\right)$ for all measurements $\mathbb{M}, \mathbb{M}^{\prime}\in \mathcal{M}(d, n)$.

\subsection{Resource monotones}
In this subsection, we will recall some resource monotones including generalized robustness and weight of states, their channel versions, and the projective robustness of states.

The generalized robustness and weight of states are defined as follows \cite{Steiner03,Datta09,Lewenstein98}:
\begin{equation}
\begin{aligned}
R_{\mathbb{F}}(\rho)=\inf\{\lambda \ | \ \rho\leq \lambda\sigma, \sigma\in \mathbb{F}\},\\
W_{\mathbb{F}}(\rho)=\sup\{\mu \ | \ \rho\geq \mu\sigma, \sigma\in \mathbb{F}\}.
\end{aligned}
\end{equation}
Using the corresponding Choi matrices, the quantum channel versions of both measures have been obtained \cite{Takagi and Regula 19, Uola20, Wilde et.al 21}
\begin{equation}
\begin{aligned}
R_{\mathbb{O}}(\mathcal{E})=\inf\{\lambda \ | \ \mathcal{J_E}\leq \lambda \mathcal{J_M}, \mathcal{M}\in \mathbb{O}\},\\
W_{\mathbb{O}}(\mathcal{E})=\sup\{\mu \ | \ \mathcal{J_E}\geq \lambda \mathcal{J_M}, \mathcal{M}\in \mathbb{O}\}.
\end{aligned}
\end{equation}
Next, we introduce the non-logarithmic variant of the max-relative entropy $D_{\max}$ \cite{Datta09}, which is defined as
\begin{equation}
R_{\max}(\rho\parallel\sigma)=2^{D_{\max}(\rho\parallel\sigma)}=\inf \{\lambda \ | \ \rho\leq\lambda\sigma\}.
\end{equation}
For any channel $\mathcal{E}, \mathcal{F}: A\rightarrow B$, the channel version of $R_{\max}$ is defined as \cite{Leditzky et.al 18, Wilde et.al 21}
\begin{equation}
\begin{aligned}
R_{max}(\mathcal{E}\parallel\mathcal{F})&=\max_{\psi_{RA}}R_{\max}(id\otimes\mathcal{E}(\psi_{RA})\parallel id\otimes\mathcal{F}(\psi_{RA}))\\
&=R_{\max}(\mathcal{J_E}\parallel \mathcal{J_F})\\
&=\inf \left\{\lambda: \mathcal{E}\leq\lambda\mathcal{F}\right\},
\end{aligned}
\end{equation}
where the inequality $\mathcal{E}\leq\lambda\mathcal{F}$ is understood as completely-positive (CP) ordering of super operators, i.e., $\lambda\mathcal{F}-\mathcal{E}$ is a CP map. Therefore, the generalized robustness and weight of quantum channels can be equivalently expressed as \cite{Regula and Takagi21}
\begin{equation}
\begin{aligned}
R_{\mathbb{O}}(\mathcal{E})=\min_{\mathcal{M}\in \mathbb{O}} R_{\max}(\mathcal{J_E}\parallel\mathcal{J_M}),\\
W_{\mathbb{O}}(\mathcal{E})^{-1}=\min_{\mathcal{M}\in \mathbb{O}} R_{\max}(\mathcal{J_M}\parallel\mathcal{J_E}).\\
\end{aligned}
\end{equation}
Recently, the projective robustness of quantum states based on the Hilbert projective metric is defined as \cite{Regula 21}
\begin{equation}
\begin{aligned}
\Omega_{\mathbb{F}}(\rho)&=\min_{\sigma\in \mathbb{F}}R_{\max}(\rho\parallel \sigma)R_{\max}(\sigma\parallel \rho),
\end{aligned}
\end{equation}
$\Omega_{\mathbb{F}}(\rho)<\infty$ if and only if there exists a free state $\sigma\in \mathbb{F}$ such that $\text{supp}\rho=\text{supp}\sigma$. Moreover, Eq.(6) has the following equivalent forms,
\begin{equation}
\begin{aligned}
\Omega_{\mathbb{F}}(\rho)&=\inf\{\lambda\mu \mid \rho\leq \lambda\sigma,\sigma\leq \mu\rho, \sigma\in \mathbb{F} \}\\
&=\inf\{\gamma \in \mathbb{R}_+\mid \rho\leq \widetilde{\sigma}\leq \gamma \rho, \widetilde{\sigma}\in \operatorname{cone}(\mathbb{F}) \}\\
&=\sup\left\{\frac{\langle A,\rho\rangle}{\langle B,\rho\rangle} \ \Big| \ \frac{\langle A,\sigma\rangle}{\langle B,\sigma\rangle}\leq 1, \forall\sigma \in \mathbb{F}, A, B \geq 0\right\},
\end{aligned}
\end{equation}
where $\operatorname{cone}(\mathbb{F})=\left\{\lambda\sigma \mid \lambda\in \mathbb{R}_+, \sigma\in\mathbb{F}\right\}$. Note that the closedness of $\mathbb{F}$ ensures that the infimum is achieved as long as it is finite.

\subsection{State discrimination and exclusion tasks}
\textit{State discrimination task}: Let us consider a scenario in which an ensemble of quantum states is given as $\{p_i,\sigma_i\}$, then we perform a chosen POVM measurement $\{M_i\}$ to distinguish which of the states was sent. For a given state ensemble, and measurement, the average probability of successfully discriminating the states is given by
\begin{equation}
p_{\text {succ }}\left(\left\{p_{i}, \sigma_{i}\right\},\left\{M_{i}\right\}\right)=\sum_{i} p_{i} \left\langle M_{i},\sigma_i \right\rangle,
\end{equation}
where $\sum_{i} p_{i}=1$ with $p_i\geq 0$ and $\sum_{i} M_{i}=\mathbb{I}$ with operators $M_i\geq0$. In this standard state discrimination task, our aim is to choose a measurement $\{M_i\}$ to maximize  Eq.(8). The generalized robustness of measurements is known to quantify the advantage that any resourceful measurement can provide over all free measurements \cite{Takagi and Regula 19,Oszmaniec19b}.

\textit{State exclusion task}: \cite{Bandyopadhyay14}
The state exclusion task can be seen as being the opposite of state discrimination. Our goal is to perform a given measurement $\{M_i\}_i$ and outputs a guess $g$ of a state that was not sent. Then, Eq.(8) is understood as the average probability of guessing incorrectly, and our goal to minimize this quantity. In state exclusion, it is the weight of measurements which quantifies the advantage \cite{Ducuara20,Uola20,Ye et.al 21}.
\section{Projective robustness of channels}
Quantum channels as the extension of quantum states play a significant vole in quantum information processing. In fact, they also can be seen as resourceful, and their operational characterization is currently an active research area \cite{Seddon19,Wang19,Saxena20,Takagi20,Regula and Takagi21,Regula and Takagi21b,Gour19a,Gour and Winter19,Liu20,Wilde et.al 21}. In this section, we will introduce the projective robustness of channels to quantify the resource amount contained a given channel.
\begin{definition}
For a given channel $\mathcal{E}$, its projective robustness is defined as
\begin{equation}
\begin{aligned}
\Omega_{\mathbb{O}}(\mathcal{E})&=\min_{\mathcal{M}\in \mathbb{O}}R_{\max}(\mathcal{E}\parallel\mathcal{M})R_{\max}(\mathcal{M}\parallel\mathcal{E})\\
&=\min_{\mathcal{M}\in \mathbb{O}} R_{\max}(\mathcal{J_E}\parallel\mathcal{J_M})R_{\max}(\mathcal{J_M}\parallel\mathcal{J_E}).
\end{aligned}
\end{equation}
\end{definition}
The first equality in Eq.(9) can be equivalently understood as the optimization over free channels $\mathcal{M}$ such that $\lambda\mathcal{M}-\mathcal{E}$ and $\mu\mathcal{E}-\mathcal{M}$ are CP maps.

\begin{theorem}
The projective robustness of channels $\Omega_{\mathbb{O}}(\mathcal{E})$ satisfies the following properties:\\
(i) $\Omega_{\mathbb{O}}(\mathcal{E})$ is finite if and only if there exists a free channel $\mathcal{M}\in \mathbb{O}$ such that $supp(\mathcal{J_M})=supp(\mathcal{J_E})$.\\
(ii) $\Omega_{\mathbb{O}}(k\mathcal{E})=\Omega_{\mathbb{O}}(\mathcal{E})$ for any $k>0$.\\
(iii) When $\mathbb{O}$ is convex, $\Omega_{\mathbb{O}}$ is quasiconvex: for any $t\in[0,1]$, it holds that
\begin{equation}
\Omega_{\mathbb{O}}(t\mathcal{E}+(1-t)\mathcal{F})\leq \max \{\Omega_{\mathbb{O}}(\mathcal{E}),\Omega_{\mathbb{O}}(\mathcal{F})\}.
\end{equation}\\
(iv) For any superchannel $\Theta\in \mathbb{S}$, it holds that
\begin{equation}
\Omega_{\mathbb{O}}(\Theta(\mathcal{E}))\leq \Omega_{\mathbb{O}}(\mathcal{E}).
\end{equation}
(v) When $\mathbb{O}$ is a compact convex set, $\Omega_{\mathbb{O}}$ can be computed as the optimal value of a conic linear optimization problem:
\begin{equation}
\begin{aligned}
&\Omega_{\mathbb{O}}(\mathcal{E})=\inf\{\gamma \in \mathbb{R}_+\mid \mathcal{J_E}\leq \mathcal{J}_{\widetilde{\mathcal{M}}}\leq \gamma \mathcal{J_E}, \mathcal{J}_{\widetilde{\mathcal{M}}}\in \operatorname{cone}(\mathbb{O}^J) \}\\
&\quad=\sup\{\langle A,\mathcal{J_E}\rangle \mid \langle B,\mathcal{J_E}\rangle=1, B-A\in \operatorname{cone}(\mathbb{O}^J)^*, A, B \geq 0\}\\
&\quad=\sup\left\{\frac{\langle A,\mathcal{J_E}\rangle}{\langle B,\mathcal{J_E}\rangle} \ \Big| \ \frac{\langle A,\mathcal{J_M}\rangle}{\langle B,\mathcal{J_M}\rangle}\leq 1, \forall\mathcal{J_M} \in \mathbb{O}^J, A, B \geq 0\right\},
\end{aligned}
\end{equation}
where $\mathbb{O}^J$ represent the set of Choi matrices corresponding to channels in $\mathbb{O}$. The $\operatorname{cone}(\mathbb{O}^J)$ is the closed convex cone generated by the set $\mathbb{O}^J$, and $\operatorname{cone}(\mathbb{O}^J)^*=\{X \mid \langle X,\mathcal{J_M}\rangle\geq 0, \forall\mathcal{J_M} \in \mathbb{O}^J\}$. Note that the compactness of $\mathbb{O}$ leads to the fact that the infimum is achieved as long as it is finite. \\
(vi) $\Omega_{\mathbb{O}}$ is lower semicontinuous, i.e. $\Omega_{\mathbb{O}}(\mathcal{E})\leq \mathop{\lim  \inf}\limits_{n\rightarrow \infty}\Omega_{\mathbb{O}}(\mathcal{E}_n)$ for any sequence $\{\mathcal{E}_n\}_n\rightarrow \mathcal{E}$.\\
(vii) It can be bounded as
\begin{equation}
\begin{aligned}
R_{\mathbb{O}}(\mathcal{E})W_{\mathbb{O}}(\mathcal{E})^{-1}\leq&\Omega_{\mathbb{O}}(\mathcal{E})\leq R_{\mathbb{O}}(\mathcal{E})R_{\max}(\mathcal{M}^*_R\parallel \mathcal{E})\\
&\Omega_{\mathbb{O}}(\mathcal{E})\leq W_{\mathbb{O}}(\mathcal{E})^{-1}R_{\max}(\mathcal{E}\parallel \mathcal{M}^*_W),
\end{aligned}
\end{equation}
where $\mathcal{M}_R^*$ is an optimal channel such that $R_{\mathbb{O}}(\mathcal{E})=R_{\max}(\mathcal{E}\parallel\mathcal{M}^*_R)$, and similarly $\mathcal{M}_W^*$ is an optimal channel such that $W_{\mathbb{O}}(\mathcal{E})=R_{\max}(\mathcal{M}^*_W\parallel\mathcal{E})^{-1}$, whenever such channels exist.\\

\begin{proof}
(i) It immediately is verified by the fact that the quantity $R_{\max}(\rho\parallel \sigma)$ is finite if and only if $\text{supp}\rho\subseteq \text{supp}\sigma$ \cite{Datta09}.

(ii) It follows from the fact that $R_{\max}(\lambda\mathcal{E}\parallel \mu\mathcal{M})=\lambda\mu^{-1}R_{\max}(\mathcal{E}\parallel \mathcal{M})$ which can be easily verified by the definition of $R_{\max}$.

(iii) If either $\Omega_{\mathbb{O}}(\mathcal{E})$ or $\Omega_{\mathbb{O}}(\mathcal{F})$ is infinite, then the relation is trivial, so assume otherwise.  Let $\mathcal{M}$ be an optimal channel such that $\mathcal{J_E}\leq \lambda \mathcal{J_M}$ and $\mathcal{J_M}\leq \mu \mathcal{J_E}$ with $\lambda\mu=\Omega_{\mathbb{O}}(\mathcal{E})$, and analogously let $\mathcal{M}'$ be an optimal channel such that $\mathcal{J_F}\leq \lambda' \mathcal{J_{M'}}$ and $\mathcal{J_{M'}}\leq \mu' \mathcal{J_F}$ with $\lambda'\mu'=\Omega_{\mathbb{O}}(\mathcal{F})$, then
\begin{equation}
\begin{aligned}
t \mathcal{J_E}+(1-t) \mathcal{J_F} &\leq t \lambda \mathcal{J_M}+(1-t) \lambda' \mathcal{J_{M'}}\\
&=[t \lambda+(1-t) \lambda'] \frac{t \lambda \mathcal{J_M}+(1-t) \lambda' \mathcal{J_{M'}}}{t \lambda+(1-t) \lambda^{\prime}},
\end{aligned}
\end{equation}
where $\frac{t \lambda \mathcal{J_M}+(1-t) \lambda' \mathcal{J_{M'}}}{t \lambda+(1-t) \lambda^{\prime}}\in \mathbb{O}^J$ by convexity of $\mathbb{O}$. Then
\begin{equation}
\begin{aligned}
\frac{t \lambda \mathcal{J_M}+(1-t) \lambda' \mathcal{J_{M'}}}{t \lambda+(1-t) \lambda^{\prime}}&\leq \frac{t \lambda \mu\mathcal{J_E}+(1-t) \lambda' \mu'\mathcal{J_{F}}}{t \lambda+(1-t) \lambda^{\prime}}\\
&\leq \frac{\max\{\lambda\mu,\lambda'\mu'\}[t\mathcal{J_E}+(1-t) \mathcal{J_{F}}]}{t \lambda+(1-t) \lambda^{\prime}}.
\end{aligned}
\end{equation}
It is easy to see that this Choi matrices is a feasible solution for $\Omega_{\mathbb{O}}(t\mathcal{E}+(1-t)\mathcal{F})$, then
\begin{equation}
\begin{aligned}
\Omega_{\mathbb{O}}(t\mathcal{E}+(1-t)\mathcal{F})&\leq [t\lambda+(1-t)\lambda']\frac{\max\{\lambda\mu,\lambda'\mu'\}}{t \lambda+(1-t) \lambda^{\prime}}\\
&=\max\{\lambda\mu,\lambda'\mu'\}.
\end{aligned}
\end{equation}
It is shown that Eq.(10) holds.

(iv) If there is no $\mathcal{M}\in \mathbb{O}$ such that $\mathrm{supp} \mathcal{J_E}=\mathrm{supp} \mathcal{J_M}$, then $\Omega_{\mathbb{O}}(\mathcal{E})=\infty$ and the result is trivial, so we shall assume
otherwise. Then let $\mathcal{M}\in \mathbb{O}$ be any channel such that $\mathcal{J_E}\leq \lambda \mathcal{J_M}$ and $\mathcal{J_M}\leq \mu \mathcal{J_E}$ with $\lambda\mu=\Omega_{\mathbb{O}}(\mathcal{E})$. By the definition of $\Theta\in \mathbb{S}$, we have that $\Theta(\mathcal{M})=\mathcal{M}'\in \mathbb{O}$. Since $\Theta$ preserves positivity, it holds that
$\mathcal{J}_{\Theta(\mathcal{E})}\leq\lambda\mathcal{J}_{\Theta(\mathcal{M})}=\lambda\mathcal{J}_{\mathcal{M}'}$ and $\mathcal{J}_{\mathcal{M}'}=\mathcal{J}_{\Theta(\mathcal{M})}\leq \mu \mathcal{J}_{\Theta(\mathcal{E})}$, so $\mathcal{M}'$ is a valid feasible solution for $\Omega_{\mathbb{O}}(\Theta(\mathcal{E}))$, which concludes the proof.

(v) From the positivity of $R_{\max}(\mathcal{E}\parallel \mathcal{M})$ for any channel $\mathcal{E}$ and $\mathcal{M}$, we have
\begin{equation}
\begin{aligned}
\Omega_{\mathbb{O}}(\mathcal{E})&=\min_{\mathcal{M}\in \mathbb{O}}\left[ \inf\{\lambda \mid \mathcal{J_E}\leq \lambda\mathcal{J_M}\}
\inf\{\mu \mid \mathcal{J_M}\leq \mu\mathcal{J_E}\}\right]\\
&=\inf\{\lambda\mu \mid \mathcal{J_E}\leq \lambda\mathcal{J_M}, \mathcal{J_M}\leq \mu\mathcal{J_E},\mathcal{M}\in \mathbb{O}\}.
\end{aligned}
\end{equation}
Observe that any feasible solution to the problem
\begin{equation}
\inf\{\gamma \mid \mathcal{J_E}\leq \mathcal{J}_{\widetilde{\mathcal{M}}}\leq \gamma\mathcal{J_E}, \mathcal{J}_{\widetilde{\mathcal{M}}}\in cone(\mathbb{O}^J)\}
\end{equation}
gives a feasible solution to Eq.(17) as
\begin{equation}
\mathcal{J_M}=\frac{\mathcal{J}_{\widetilde{\mathcal{M}}}}{\langle \mathbb{I},\mathcal{J}_{\widetilde{\mathcal{M}}}\rangle}, \lambda=\langle \mathbb{I},\mathcal{J}_{\widetilde{\mathcal{M}}}\rangle, \mu=\frac{\gamma}{\langle \mathbb{I},\mathcal{J}_{\widetilde{\mathcal{M}}}\rangle}
\end{equation}
with objective function value $\lambda\mu=\gamma$. Conversely, any feasible solution $\{\mathcal{J_M}, \lambda, \mu\}$ to Eq.(17) gives a feasible solution to Eq.(18) as $\mathcal{J}_{\widetilde{\mathcal{M}}}=\lambda\mathcal{J}_{\mathcal{M}}, \gamma=\lambda\mu$. Thus, the two problems are equivalent.\\
Writing the Lagrangian as
\begin{equation}
\begin{aligned}
&L(\gamma,\mathcal{J}_{\widetilde{\mathcal{M}}};A,B,C)\\
&\quad=\gamma-\langle A, \mathcal{J}_{\widetilde{\mathcal{M}}}-\mathcal{J_E}\rangle-\langle B, \gamma\mathcal{J_E}-\mathcal{J}_{\widetilde{\mathcal{M}}}\rangle-\langle C,\mathcal{J}_{\widetilde{\mathcal{M}}}\rangle\\
&\quad=\gamma(1-\langle B,\mathcal{J_E}\rangle)+\langle B-A-C, \mathcal{J}_{\widetilde{\mathcal{M}}}\rangle+\langle A,\mathcal{J_E}\rangle.
\end{aligned}
\end{equation}
Optimizing over the Lagrange multipliers $A, B\geq 0$ and $C\in \operatorname{cone}(\mathbb{O}^J)^*$, the corresponding dual form of primal problem Eq.(18) is written as
\begin{equation}
\begin{aligned}
\sup\{\langle A, \mathcal{J_E}\rangle \mid \langle B,\mathcal{J_E}\rangle=1, B-A\in \operatorname{cone}(\mathbb{O}^J)^*, A, B\geq 0\}
\\=\sup\left\{\frac{\langle A, \mathcal{J_E}\rangle}{\langle B,\mathcal{J_E}\rangle} \ \Big| \  \frac{\langle A, \mathcal{J_M}\rangle}{\langle B,\mathcal{J_M}\rangle}\leq 1, \forall \mathcal{J_M}\in \mathbb{O}^J, A, B\geq 0\right\}.
\end{aligned}
\end{equation}
If we take $B=\mathbb{I}$ and $A=\epsilon\mathbb{I}$ for $0<\epsilon<1$, it is obvious that $A$ and $B$ are strictly feasible for the dual. Thus, it follows from Slater's theorem \cite{Boyd04} that strong duality holds, the optimal value of the primal problem Eq.(18) is equal to that of dual problem Eq.(21).\\
The second line of Eq.(21) follows since any feasible solution to this program can be rescaled
as $A\mapsto\frac{A}{\langle B, \mathcal{J_E}\rangle}, B\mapsto\frac{B}{\langle B, \mathcal{J_E}\rangle}$ to give a feasible solution to the dual, and vice versa. Here, we implicitly constrain ourselves to $B$ such that $\langle B, \mathcal{J_E}\rangle\neq 0$ and $\langle B, \mathcal{J_M}\rangle\neq 0, \forall \mathcal{J_M}\in \mathbb{O}^J$; this can always be achieved by taking $B$ as a small multiple of the identity.

(vi) Verifying lower semicontinuity of $\Omega_{\mathbb{O}}$ is equivalent to showing that the sublevel sets
$\left\{\mathcal{E} \mid \Omega_{\mathbb{O}}(\mathcal{E})\leq \gamma\right\}$
are closed for all $\gamma\in \mathbb{R}$ \cite{Rockafellar70}. Consider then a sequence
$\mathcal{E}_{n} \stackrel{n \rightarrow \infty}{\longrightarrow} \mathcal{E}$ such that $\Omega_{\mathbb{O}}\left(\mathcal{E}_{n}\right) \leq \gamma ~\forall n$ for some $\gamma$, where we can take $\gamma\geq1$ to avoid trivial cases.  By (v), this entails that there exists $\mathcal{J}_{\widetilde{\mathcal{M}}_n}\in \operatorname{cone}(\mathbb{O}^J)$ such that $\mathcal{J}_{\mathcal{E}_n}\leq \mathcal{J}_{\widetilde{\mathcal{M}}_n}\leq \gamma\mathcal{J}_{\mathcal{E}_n}$ for each $n$. Since $\{\mathcal{J}_{\widetilde{\mathcal{M}}_n}\}$ forms a bounded sequence, by the Bolzano-Weierstrass theorem we can
assume that it converges as $\{\mathcal{J}_{\widetilde{\mathcal{M}}_n}\}\rightarrow \mathcal{J}_{\widetilde{\mathcal{M}}}$, up to passing to a subsequence. The closedness of $\operatorname{cone}(\mathbb{O}^J)$ ensures that $\mathcal{J}_{\widetilde{\mathcal{M}}}\in \operatorname{cone}(\mathbb{O}^J)$. The convergent sequences $\{\mathcal{J}_{\widetilde{\mathcal{M}}_n}-\mathcal{J}_{\mathcal{E}_n}\}$ and
$\{\gamma\mathcal{J}_{\mathcal{E}_n}-\mathcal{J}_{\widetilde{\mathcal{M}}_n}\}$ then must converge to positive semidefinite operators by the closedness of the positive semidefinite cone. This gives $\mathcal{J}_{\mathcal{E}}\leq \mathcal{J}_{\widetilde{\mathcal{M}}}\leq \gamma\mathcal{J}_{\mathcal{E}}$, showing that the sublevel set of $\Omega_{\mathbb{O}}$ is closed. Since $\gamma$ was arbitrary, the desired property is proved.

(vii) The lower bound is obtained by noting that
\begin{equation}
\begin{aligned}
\Omega_{\mathbb{O}}(\mathcal{E}) &=\min _{\mathcal{M} \in \mathbb{O}} R_{\max }(\mathcal{E} \| \mathcal{M}) R_{\max }(\mathcal{M} \| \mathcal{E}) \\
& \geq\left[\min _{\mathcal{M} \in \mathbb{O}} R_{\max }(\mathcal{E} \| \mathcal{M})\right]\left[\min _{\mathcal{M} \in \mathbb{O}} R_{\max }(\mathcal{M} \| \mathcal{E})\right] \\
&=R_{\mathbb{O}}(\mathcal{E}) W_{\mathbb{O}}(\mathcal{E})^{-1}.
\end{aligned}
\end{equation}
The upper bounds follow by using $\mathcal{M}^*_{R}$ and $\mathcal{M}^*_{W}$ as feasible solutions in the definition of $\Omega_{\mathbb{O}}$.

\end{proof}
\end{theorem}

\begin{proposition}
 For any replacement channel $\mathcal{R_\omega}: \mathbb{L}(A)\rightarrow \mathbb{L}(B)$ defined as $R_\omega(\cdot)=\operatorname{Tr}(\cdot)\omega$ with some fixed $\omega\in \mathbb{D}(B)$, it holds that
\begin{equation}
\begin{aligned}
\Omega_{\mathbb{O}}(\mathcal{R_\omega})\geq \Omega_{\mathbb{F}}(\mathcal{\omega}).
\end{aligned}
\end{equation}
If the class of operations $\mathbb{O}$ contains all replacement channels $\mathcal{R_\sigma}$ with $\sigma\in \mathbb{F}$, then equality holds in the above.
\end{proposition}
\begin{proof}
Taking any dual feasible solution $A,B$ for $\Omega_{\mathbb{F}}(\omega)$ and any $\tau\in \mathbb{F}$, we can see that $\langle \tau^T\otimes A, \mathcal{J_M}\rangle=\langle A,\mathcal{M}(\tau)\rangle$, $\langle \tau^T\otimes B, \mathcal{J_M}\rangle=\langle B,\mathcal{M}(\tau)\rangle$. Since
\begin{equation}
\begin{aligned}
\frac{\langle \tau^T\otimes A, \mathcal{J_M}\rangle}{\langle \tau^T\otimes B, \mathcal{J_M}\rangle}=\frac{\langle A,\mathcal{M}(\tau)\rangle}{\langle B,\mathcal{M}(\tau)\rangle}\leq 1
\end{aligned}
\end{equation}
for any $\mathcal{M}\in \mathbb{O}$ using the Choi-Jamio${\l}$kowski isomorphism, it implies that $\tau^T\otimes A$ and $\tau^T\otimes B$ is dual feasible to $\Omega_{\mathbb{O}}(\mathcal{R_\omega})$. This immediately gives that
\begin{equation}
\begin{aligned}
\Omega_{\mathbb{O}}(\mathcal{R_\omega})&\geq \sup\left\{\frac{\langle \tau^T\otimes A, \mathcal{J_{\mathcal{R_\omega}}}\rangle}{\langle \tau^T\otimes B, \mathcal{J_{\mathcal{R_\omega}}}\rangle} \ \Big| \  \frac{\langle A,\sigma\rangle}{\langle B,\sigma\rangle}\leq 1, \forall \sigma\in \mathbb{F}, A,B\geq 0\right\}\\
&=\sup\left\{\frac{\langle A, \omega\rangle}{\langle B, \omega \rangle} \ \Big| \  \frac{\langle A,\sigma\rangle}{\langle B,\sigma\rangle}\leq 1, \forall \sigma\in \mathbb{F}, A,B\geq 0\right\}\\
&=\Omega_{\mathbb{F}}(\omega).
\end{aligned}
\end{equation}
Now, assume that $\mathcal{R}_{\sigma}\in \mathbb{O}$ for any $\sigma\in \mathbb{F}$. From the definition of the projective robustness of states, let $\sigma\in\mathbb{F}$ be an optimal state such that $\omega\leq \lambda\sigma$ and $\sigma\leq\mu \omega$ with $\lambda\mu=\Omega_{\mathbb{F}}(\omega)$. Moreover, $\mathbb{I}\otimes\omega\leq\mathbb{I}\otimes(\lambda\sigma)=\lambda\mathcal{J_{R_\sigma}}$ and $\mathbb{I}\otimes\sigma\leq\mathbb{I}\otimes(\mu\omega)=\mu\mathcal{J_{R_\omega}}$, which imply $\Omega_{\mathbb{O}}(\mathcal{R_\omega})\leq\Omega_{\mathbb{F}}(\omega)$.
\end{proof}
The above proposition explicitly shows the relation between the projective robustness of states and its channel version.  As a consequence of part (iv) of Theorem 1 and Proposition 1, we have the following corollary.

\begin{corollary}
Let $\mathbb{O}$ be any class of free operations, which satisfies $\mathcal{R}_\sigma\in \mathbb{O} \ \forall \sigma\in\mathbb{F}$ . Let $\mathcal{N}:\mathbb{L}(A)\rightarrow \mathbb{L}(B)$ be any channel such that:

(i) there exists a free superchannel $\Gamma\in \mathbb{S}$ and a state $\omega$ such that $\Gamma(\mathcal{R}_\omega)=\mathcal{N}$,

(ii) there exists a free superchannel $\Theta\in \mathbb{S}$ such that $\Theta(\mathcal{N})=\mathcal{R}_\omega$.\\
Then
\begin{equation}
\Omega_{\mathbb{O}}(\mathcal{N})=\Omega_{\mathbb{F}}(\omega).
\end{equation}
\end{corollary}
By the projective robustness of channels, Corollary 1 again claims a universal property given in Ref.\cite{Regula and Takagi21}, which states that if some channels can be reversibly interconverted with state resources through free superchannels, then we regard these two types of resources as equivalent.

In this section, we study the properties of the projective robustness of channels in detail, which will lay a good foundation for further research to manipulate channel resources.

\subsection{Channel resource distillation}

Resource distillation is one of the most important operational tasks in resource theory. The dynamical resource distillation can be generally understood as converting a noisy resource channel $\mathcal{E}: \mathbb{L}(A) \rightarrow \mathbb{L}(B)$ into some target channels $\mathcal{T}: \mathbb{L}(C) \rightarrow \mathbb{L}(D)$, which are viewed as “pure” or “perfect” resources.  Note that we choose 
some resourceful channels converting any pure state into the pure state as target channel $\mathcal{T}$.
Our purpose is to understand when the transformation $\Theta(\mathcal{E})=\mathcal{T}$ can be achieved by using the free superchannel $\Theta \in \mathbb{S}$. 

However, the practical physical system cannot perfectly transform the channels, i.e., the manipulation of channels always leads to the possibility of error of the transformation. Thus, the core task of resource distillation is to provide a threshold for such error to characterize capabilities and limitations of distillation protocols manipulating quantum
resources. To accomplish this task, we start by introducing the worst-case fidelity between the two channels $\mathcal{E}, \mathcal{F}: \mathbb{L}(A) \rightarrow \mathbb{L}(B)$, which is defined as follows \cite{Belavkin 05,Gilchrist 05} 
\begin{equation}
\begin{aligned}
F(\mathcal{E}, \mathcal{F}): &=\min _{\rho_{R A}} F(\mathrm{id} \otimes \mathcal{E}(\rho), \mathrm{id} \otimes \mathcal{F}(\rho)) \\
&=\min _{\psi_{R A}} F(\mathrm{id} \otimes \mathcal{E}(\psi), \mathrm{id} \otimes \mathcal{F}(\psi))
\end{aligned}
\end{equation}
where $F(\rho, \sigma)=\|\sqrt{\rho} \sqrt{\sigma}\|_{1}^{2}$ is the fidelity, and in the second line the optimization is over all pure input states. Thus, our goal is to achieve a free superchannel satisfying $F(\Theta(\mathcal{E}), \mathcal{T}) \geq 1-\varepsilon$ for some small error $\varepsilon>0$. To quantify how closely the target channel is to a free channel, the fidelity-based measure of the overlap with free channels is defined as \cite{Regula and Takagi21}
\begin{equation}
\begin{aligned}
F_{\mathbb{O}}(\mathcal{T})&=\max _{\mathcal{M} \in \mathbb{O}} F(\mathcal{T}, \mathcal{M})\\
&= \max _{\mathcal{M} \in \mathbb{O}}\min _{\psi_{R A}}F(\mathrm{id} \otimes \mathcal{T}(\psi), \mathrm{id} \otimes \mathcal{M}(\psi))\\
&= \min _{\psi_{R A}}\max _{\mathcal{M} \in \mathbb{O}}\langle\mathrm{id} \otimes \mathcal{T}(\psi), \mathrm{id} \otimes \mathcal{M}(\psi)\rangle.
\end{aligned}
\end{equation}
In the state case, we can write
\begin{equation}
\begin{aligned}
F_{\mathbb{F}}(\phi)&=\max _{\sigma \in \mathbb{F}} F(\phi, \sigma)=\max _{\sigma \in \mathbb{F}} \langle \phi, \sigma\rangle.
\end{aligned}
\end{equation}

The following result builds a bound on distillation error in channel deterministic transformations.
\begin{theorem}
Let $\mathcal{T}\in \operatorname{CPTP}$ and suppose that $id\otimes \mathcal{T}(\psi)$ is pure for any pure state $\psi$. If there exists a free superchannel $\Theta\in \mathbb{S}$ such that $\Theta(\mathcal{E})=\mathcal{N}$ with $F(\mathcal{N},\mathcal{T})\geq1-\varepsilon$, then
\begin{equation}\label{distillerror1}
\varepsilon\geq\left(\frac{F_{\mathbb{O}}(\mathcal{T})}{1-F_{\mathbb{O}}(\mathcal{T})}\Omega_{\mathbb{O}}(\mathcal{E})+1\right)^{-1}.
\end{equation}
\end{theorem}
\begin{proof}
The relation is trivial when $\Omega_{\mathbb{O}}(\mathcal{E})=\infty$, so assume otherwise. Then $\Omega_{\mathbb{O}}(\mathcal{N})<\infty$, so let $\mathcal{M}\in \mathbb{O}$  be a free channel such that $\mathcal{J}_{\mathcal{N}}\leq \lambda \mathcal{J}_{\mathcal{M}}$ and $\mathcal{J}_{\mathcal{M}}\leq \mu \mathcal{J}_{\mathcal{N}}$ with $\lambda\mu=\Omega_{\mathbb{O}}(\mathcal{N})$. For any $\psi$, we have
\begin{equation}
\begin{aligned}
1-\varepsilon & \leq \langle \mathrm{id}\otimes \mathcal{N}(\psi), \mathrm{id} \otimes \mathcal{T}(\psi)\rangle\\
&\leq \lambda\langle \mathrm{id}\otimes \mathcal{M}(\psi), \mathrm{id} \otimes \mathcal{T}(\psi)\rangle\\
&\leq \lambda F^{\psi}_{\mathbb{O}}(\mathcal{T})
\end{aligned}
\end{equation}
where $F^{\psi}_{\mathbb{O}}(\mathcal{T})=\max_{\mathcal{M}\in\mathbb{O}}\langle \mathrm{id}\otimes \mathcal{M}(\psi), \mathrm{id} \otimes \mathcal{T}(\psi)\rangle$.
From the second and third lines of Eq.(31), we have
\begin{equation}
\begin{aligned}
1-F^{\psi}_{\mathbb{O}}(\mathcal{T}) &\leq 1-\langle \mathrm{id}\otimes \mathcal{M}(\psi), \mathrm{id} \otimes \mathcal{T}(\psi)\rangle\\
&=\left\langle \mathrm{id}\otimes \tilde{\Theta}(\mathcal{M})(\psi), \mathbb{I}-\mathrm{id} \otimes \mathcal{T}(\psi)\right\rangle\\
&\leq \mu \left\langle \mathrm{id}\otimes \mathcal{N}(\psi), \mathbb{I}-\mathrm{id} \otimes \mathcal{T}(\psi)\right\rangle\\
&\leq \mu \varepsilon.
\end{aligned}
\end{equation}
Due to the arbitrariness of $\psi$, for Eqs.(31) and (32), it immediately holds that
\begin{equation}
1-\varepsilon \leq  \lambda F_{\mathbb{O}}(\mathcal{T})
\end{equation}
and
\begin{equation}
1-F_{\mathbb{O}}(\mathcal{T}) \leq \mu \varepsilon.
\end{equation}
From Eqs.(33) and (34),
\begin{equation}
\begin{aligned}
\Omega_{\mathbb{O}}(\mathcal{E})\geq \Omega_{\mathbb{O}}(\mathcal{N})=\lambda\mu \geq \frac{(1-\varepsilon)(1-F_{\mathbb{O}}(\mathcal{T}))}{\varepsilon F_{\mathbb{O}}(\mathcal{T})},
\end{aligned}
\end{equation}
where we use the monotonicity of $\Omega_{\mathbb{O}}$. Rearranging, we can immediately get Eq.(30).
\end{proof}
The bound  above provides a general error threshold that cannot be exceeded by any deterministic protocol in any channel resources. 

\begin{remark}
Particularly, it is easy to see that the above theorem still holds if we take $\mathcal{T}$ as unitary channel $\mathcal{U}=U\cdot U^{\dagger}$ or the replacement channel $\mathcal{R}_{\phi}=\text{Tr}(\cdot)\phi$ for some resourceful pure state $\phi$. Moreover, when the input and target channels are the preparation channels $\mathcal{P}_\rho$ and $\mathcal{P}_\phi$, respectively, we have
$\varepsilon\geq\left(\frac{F_{\mathbb{F}}(\phi)}{1-F_{\mathbb{F}}(\phi)}\Omega_{\mathbb{F}}(\rho)+1\right)^{-1}$, which is consistent with the result of deterministic state resource distillation in Ref.\cite{Regula 21}.
\end{remark}
Notice that part (iv) of Theorem 1 provides a necessary condition  for  the  deterministic  transformation  of  single-channel. Next, we will provide a necessary condition for the deterministic transformation from multiple channels to single-channel in the most general protocol, i.e, $\Omega_{\mathbb{O}}$ satisfies sub- or super-multiplicativity under any free quantum process.

{\textbf{Sub- or supermultiplicativity.}} To construct the most general protocol to manipulate multiple channels (including parallel, sequential, or adaptive, with or without a definite causal order), the set of free quantum processes was defined as follows \cite{Regula and Takagi21}
\begin{equation}
\mathbb{S}_{(n)}=\{\Upsilon \ \mid \ \Upsilon(\mathcal{M}_1,...,\mathcal{M}_n)\in \mathbb{O}, \forall \mathcal{M}_1,...,\mathcal{M}_n\in \mathbb{O}\},
\end{equation}
where the $n$-channel quantum process $\Upsilon$ is an $n$-linear map which transforms $n$-channels to a single channel and we only require $\Upsilon(\mathcal{N}_1,...,\mathcal{N}_n)\in \text{CPTP}$ for any $\mathcal{N}_1,...,\mathcal{N}_n\in \text{CPTP}$. Obviously, superchannel is a special case of quantum process $\Upsilon$ where ones take a single channel as the input.

\begin{theorem}
Let $(\mathcal{E}_1,...,\mathcal{E}_n)$ be a collection of $n$ channels. For any free protocol $\Upsilon\in \mathbb{S}_{(n)}$ it holds that
\begin{equation}
\Omega_{\mathbb{O}}(\Upsilon(\mathcal{E}_1,...,\mathcal{E}_n))\leq \prod_i\Omega_{\mathbb{O}}(\mathcal{E}_i).
\end{equation}
In particular, let the $n$-tuple $(\mathcal{E},...,\mathcal{E})=:\mathcal{E}^{\times n}$ represent an application of $n$ copies of the same channel, we have $\Omega_{\mathbb{O}}\left(\Upsilon\left(\mathcal{E}^{\times n}\right)\right) \leq \Omega_{\mathbb{O}}(\mathcal{E})^{n}$.
\end{theorem}
\begin{proof}
For each $\mathcal{E}_i$, we assume that $\Omega_{\mathbb{O}}(\mathcal{E}_i)$ is finite. Otherwise, the result is trivial.
Let $\mathcal{M}_i\in \mathbb{O}$ be an optimal channel such that $\mathcal{J}_{\mathcal{E}_i}\leq \lambda_i \mathcal{J}_{\mathcal{M}_i}$ and $\mathcal{J}_{\mathcal{M}_i}\leq \mu_i \mathcal{J}_{\mathcal{E}_i}$ with $\lambda_i\mu_i=\Omega_{\mathbb{O}}(\mathcal{E}_i)$. Using the $n$-linearity of
the transformation $\Upsilon$, we have
\begin{equation}
\begin{aligned}
&\Upsilon\left(\mathcal{E}_{1}, \ldots, \mathcal{E}_{n}\right)\\
&=\Upsilon\left(\mathcal{E}_{1}-\lambda_{1} \mathcal{M}_{1}, \mathcal{E}_{2}, \ldots, \mathcal{E}_{n}\right)+\Upsilon\left(\lambda_{1} \mathcal{M}_{1}, \mathcal{E}_{2}, \ldots, \mathcal{E}_{n}\right) \\
&=\Upsilon\left(\mathcal{E}_{1}-\lambda_{1} \mathcal{M}_{1}, \mathcal{E}_{2}, \ldots, \mathcal{E}_{n}\right)+\Upsilon\left(\lambda_{1} \mathcal{M}_{1}, \mathcal{E}_{2}-\lambda_{2} \mathcal{M}_{2}, \mathcal{E}_{3}, \ldots, \mathcal{E}_{n}\right)\\
&+\Upsilon\left(\lambda_{1} \mathcal{M}_{1}, \lambda_{2} \mathcal{M}_{2}, \mathcal{E}_{3}, \ldots, \mathcal{E}_{n}\right) \\
& \vdots \\
&=\Upsilon\left.(\mathcal{E}_{1}-\lambda_{1} \mathcal{M}_{1}, \mathcal{E}_{2}, \ldots, \mathcal{E}_{n})+\ldots+\Upsilon(\lambda_{1} \mathcal{M}_{1}, \ldots, \lambda_{n-1} \mathcal{M}_{n-1},\right.\\
&\left.\mathcal{E}_{n}-\lambda_{n} \mathcal{M}_{n})
+\Upsilon(\lambda_{1} \mathcal{M}_{1}, \ldots, \lambda_{n} \mathcal{M}_{n})\right. \\
\end{aligned}
\end{equation}
and
\begin{equation}
\begin{aligned}
&\Upsilon\left(\mathcal{M}_{1}, \ldots, \mathcal{M}_{n}\right)\\
&=\Upsilon\left(\mathcal{M}_{1}-\mu_{1} \mathcal{E}_{1}, \mathcal{M}_{2}, \ldots, \mathcal{M}_{n}\right)
+\Upsilon\left(\mu_{1} \mathcal{E}_{1}, \mathcal{M}_{2}, \ldots, \mathcal{M}_{n}\right) \\
&=\left.\Upsilon(\mathcal{M}_{1}-\mu_{1} \mathcal{E}_{1}, \mathcal{M}_{2}, \ldots, \mathcal{M}_{n})+\Upsilon(\mu_{1} \mathcal{E}_{1}, \mathcal{M}_{2}-\mu_{2} \mathcal{E}_{2}, \mathcal{M}_{3},\right.\\
&\left.\ldots, \mathcal{M}_{n})+\Upsilon(\mu_{1} \mathcal{E}_{1}, \mu_{2} \mathcal{E}_{2}, \mathcal{M}_{3}, \ldots, \mathcal{M}_{n})\right. \\
& \vdots \\
&=\left.\Upsilon(\mathcal{M}_{1}-\mu_{1} \mathcal{E}_{1}, \mathcal{M}_{2}, \ldots, \mathcal{M}_{n})+\ldots+\Upsilon(\mu_{1} \mathcal{E}_{1}, \ldots, \mu_{n-1} \mathcal{E}_{n-1},\right.\\
&\left.\mathcal{M}_{n}-\mu_{n} \mathcal{E}_{n})
+\Upsilon(\mu_{1} \mathcal{E}_{1}, \ldots, \mu_{n} \mathcal{E}_{n})\right. .
\end{aligned}
\end{equation}
By the positivity of $\Upsilon$, we have
\begin{equation}
\begin{aligned}
0 & \leq \Upsilon\left(\lambda_{1} \mathcal{M}_{1}, \ldots, \lambda_{n}\mathcal{M}_{n}\right)-\Upsilon\left(\mathcal{E}_{1}, \ldots, \mathcal{E}_{n}\right) \\
&=\left(\prod_{i} \lambda_{i}\right) \Upsilon\left(\mathcal{M}_{1}, \ldots, \mathcal{M}_{n}\right)-\Upsilon\left(\mathcal{E}_{1}, \ldots, \mathcal{E}_{n}\right)\\
&=\left(\prod_{i} \lambda_{i}\right) \mathcal{M}^{\prime}-\Upsilon\left(\mathcal{E}_{1}, \ldots, \mathcal{E}_{n}\right)
\end{aligned}
\end{equation}
and
\begin{equation}
\begin{aligned}
0 & \leq \Upsilon\left(\mu_{1}\mathcal{E}_{1}, \ldots, \mu_{n}\mathcal{E}_{n}\right)-\Upsilon\left( \mathcal{M}_{1}, \ldots, \mathcal{M}_{n}\right) \\
&=\left(\prod_{i} \mu_{i}\right)\Upsilon\left(\mathcal{E}_{1}, \ldots, \mathcal{E}_{n}\right)- \Upsilon\left(\mathcal{M}_{1}, \ldots, \mathcal{M}_{n}\right) \\
&=\left(\prod_{i} \mu_{i}\right)\Upsilon\left(\mathcal{E}_{1}, \ldots, \mathcal{E}_{n}\right)- \mathcal{M}^{\prime}
\end{aligned}
\end{equation}
for $\Upsilon\left(\mathcal{M}_{1}, \ldots, \mathcal{M}_{n}\right)=\mathcal{M}^{\prime}\in \mathbb{O}$. From Eqs.(40) and (41), it follows that $\prod_{i}\Omega_{\mathbb{O}}(\mathcal{E}_i)=\prod_{i}\lambda_i\mu_i$ is a feasible optimal value for $\Omega_{\mathbb{O}}(\Upsilon(\mathcal{E}_1,...,\mathcal{E}_n))$, which leads to Eq.(37).
\end{proof}
Since the tensor product and composition are two special cases of free quantum process, the following corollary is immediately concluded.

\begin{corollary}
 If the resource theory is well-behaved under tensor product, i.e. $\mathcal{M}, \mathcal{M}'\in \mathbb{O}\Rightarrow \mathcal{M}\otimes \mathcal{M}'\in \mathbb{O}$, then
\begin{equation}
\Omega_{\mathbb{O}}(\mathcal{E}\otimes\mathcal{F})\leq \Omega_{\mathbb{O}}(\mathcal{E})\Omega_{\mathbb{O}}(\mathcal{F}).
\end{equation}
If the resource theory is well-behaved under composition, i.e. $\mathcal{M}, \mathcal{M}'\in \mathbb{O}\Rightarrow \mathcal{M}\circ \mathcal{M}'\in \mathbb{O}$, then
\begin{equation}
\Omega_{\mathbb{O}}(\mathcal{E}\circ\mathcal{F})\leq \Omega_{\mathbb{O}}(\mathcal{E})\Omega_{\mathbb{O}}(\mathcal{F}).
\end{equation}
\end{corollary}
The universal bound of Theorem 2 and the sub- or supermultiplicativity  of  $\Omega_{\mathbb{O}}$ can provide limitations on the overhead required in the most general many-copy channel transformation protocol.
\begin{theorem}
For any distillation protocol $\Upsilon\in \mathbb{S}_{(n)}$ which transforms $n$ uses of a channel $\mathcal{E}$ to some target resourceful channel $\mathcal{T}$ up to accuracy $\varepsilon>0$, it necessarily holds that
\begin{equation}
n\geq\log_{\Omega_{\mathbb{O}}(\mathcal{E})}
\frac{(1-\varepsilon)(1-F_{\mathbb{O}}(\mathcal{T}))}{\varepsilon F_{\mathbb{O}}(\mathcal{T})}
\end{equation}
\end{theorem}
\begin{proof}
It immediately follows from Eqs.(30) and (37).
\end{proof}
\begin{remark}
Particularly, it is easy to see that the above theorem still holds if we take $\mathcal{T}$ as unitary channel $\mathcal{U}=U\cdot U^{\dagger}$ or the replacement channel $\mathcal{R}_{\phi}=\text{Tr}(\cdot)\phi$ for some resourceful pure state $\phi$. Moreover, when the input and target channels are the preparation channels $\mathcal{P}_\rho$ and $\mathcal{P}_\phi$, respectively, this result reduces to the result of deterministic state resource distillation overheads in Ref.\cite{Regula 21}, i.e.,
$n\geq\log_{\Omega_{\mathbb{F}}(\rho)}\frac{(1-\varepsilon)(1-F_{\mathbb{F}}(\phi))}{\varepsilon F_{\mathbb{F}}(\phi)}$, which characterize the number of copies of a state needed to perform the distillation $\mathcal{M}\left(\rho^{\otimes n}\right) \rightarrow \phi$ up to error $\varepsilon$.
\end{remark}
\subsection{Interpretation in state discrimination}
In this section, we will show a final application of $\Omega_{\mathbb{O}}$ in a variant of state discrimination task. Let us consider a scenario in which an ensemble of quantum states is given as $\{p_i,\sigma_i\}$ with $p_i\geq 0$ and $\sum_{i} p_{i}=1$, and a given bipartite channel $\mathbb{I}\otimes\Lambda$ is applied to every state. By measuring the output of this process with a chosen POVM $\{M_i\}$, we aim to distinguish which of the states was sent. The average probability of successfully discriminating the states is given by
\begin{equation}
p_{\text {succ }}\left(\left\{p_{i}, \sigma_{i}\right\},\left\{M_{i}\right\}, \mathbb{I}\otimes\Lambda\right)=\sum_{i} p_{i} \langle \mathbb{I}\otimes\Lambda(\sigma_i), M_{i}\rangle.
\end{equation}
When we consider state discrimination, our aim is to choose a measurement to maximize Eq.(45). The $R_{\mathbb{O}}$ is known to quantify the advantage that a given channel can provide over all free channels \cite{Takagi and Regula 19}. When we consider state exclusion, our goal is to do our best to determine which states was not sent. Then, Eq.(45) is understood as the average probability of guessing incorrectly, and our aim is to minimize this quantity. In state exclusion, it is the $W_{\mathbb{O}}$  which quantifies the advantage \cite{Uola20,Ye et.al 21}. Next, we can show that $\Omega_{\mathbb{O}}(\Lambda)$ quantifies the advantage when state discrimination and exclusion tasks are considered at the same time.

\begin{theorem}
The projective robustness of channels quantifies the maximal advantage that a given channel $\Lambda$ gives over all free channels $\Xi\in\mathbb{O}$ in simultaneous discrimination and exclusion of a fixed state ensemble, as quantified by the ratio
\begin{equation}
\frac{p_{\text{succ}}\left(\left\{p_{i}, \sigma_{i}\right\},\left\{M_{i}\right\}, \mathbb{I}\otimes\Lambda\right)}{p_{\text{succ}}\left(\left\{p_{i}, \sigma_{i}\right\},\left\{N_{i}\right\}, \mathbb{I}\otimes\Lambda\right)},
\end{equation}
where we aim to perform state discrimination with measurement $\{M_i\}$ and state exclusion with measurement $\{N_i\}$.

Specifically,
\begin{equation}
\begin{aligned}
\sup _{\left\{p_{i}, \sigma_{i}\right\} \atop\left\{M_{i}\right\},\left\{N_{i}\right\}} \frac{\frac{p_{\text {succ }}\left(\left\{p_{i}, \sigma_{i}\right\},\left\{M_{i}\right\}, \mathbb{I}\otimes\Lambda\right)}{p_{\text {succ }}\left(\left\{p_{i}, \sigma_{i}\right\},\left\{N_{i}\right\}, \mathbb{I}\otimes\Lambda\right)}}{\mathop{\max} \limits _{\Xi \in \mathbb{O}} \frac{p_{\text {succ }}\left(\left\{p_{i}, \sigma_{i}\right\},\left\{M_{i}\right\}, \mathbb{I}\otimes\Xi\right)}{p_{\text {succ }}\left(\left\{p_{i}, \sigma_{i}\right\},\left\{N_{i}\right\}, \mathbb{I}\otimes\Xi\right)}}=\Omega_{\mathbb{O}}(\Lambda),
\end{aligned}
\end{equation}
where the maximization is over all finite state ensembles $\{p_i,\sigma_i\}_{i=1}^n$ and all POVMs $\{M_i\}_{i=1}^n$, $\{N_i\}_{i=1}^n$ for which the expression in Eq. (47) is well defined.
\end{theorem}
\begin{proof}
Let us assume that $\Omega_{\mathbb{O}}(\Lambda)$ is finite and take $\mathcal{J}_{\Lambda}\leq \lambda\mathcal{J}_{\Xi}$ and $\mathcal{J}_{\Xi}\leq \mu\mathcal{J}_{\Lambda}$ such that $\Omega_{\mathbb{O}}(\Lambda)=\lambda\mu$. Let $\{p_i,\sigma_i\}$ be any state ensemble, and $\{M_i\}_{i=1}^n$, $\{N_i\}_{i=1}^n$ be any feasible POVMs. Then
\begin{equation}
\begin{aligned}
&\frac{p_{\text {succ }}\left(\left\{p_{i}, \sigma_{i}\right\},\left\{M_{i}\right\}, \mathbb{I}\otimes\Lambda\right)}{p_{\text {succ }}\left(\left\{p_{i}, \sigma_{i}\right\},\left\{N_{i}\right\}, \mathbb{I}\otimes\Lambda\right)}\\
&\qquad \leq \frac{\lambda p_{\text {succ }}\left(\left\{p_{i}, \sigma_{i}\right\},\left\{M_{i}\right\}, \mathbb{I}\otimes\Xi\right)}{\mu^{-1} p_{\text {succ }}\left(\left\{p_{i}, \sigma_{i}\right\},\left\{N_{i}\right\}, \mathbb{I}\otimes\Xi\right)} \\
&\qquad \leq \lambda \mu \max _{\Xi \in \mathbb{O}} \frac{ p_{\text {succ }}\left(\left\{p_{i}, \sigma_{i}\right\},\left\{M_{i}\right\}, \mathbb{I}\otimes\Xi\right)}{p_{\text {succ }}\left(\left\{p_{i}, \sigma_{i}\right\},\left\{N_{i}\right\}, \mathbb{I}\otimes\Xi\right)},
\end{aligned}
\end{equation}
where the first inequality follows by the complete positivity of $\Lambda$ and linearity of $p_{\text{succ}}$. This implies that $\Omega_{\mathbb{O}}(\Lambda)$ always upper bounds the left hand side(LHS) of Eq.(47). Next, we choose two POVMs $\{M_i\}_{i=1}^2=\left\{\frac{A}{\|A\|_{\infty}}, \mathbb{I}-\frac{A}{\|A\|_{\infty}}\right\}$ and $\{N_i\}_{i=1}^2=\left\{\frac{B}{\|B\|_{\infty}}, \mathbb{I}-\frac{B}{\|B\|_{\infty}}\right\}$ where $A$ and $B$ are any dual feasible solutions for $\Omega_{\mathbb{O}}(\Lambda)$, and define the state ensemble $\{p_i,\sigma_i\}_{i=1}^2$ as $p_1=1$, $\sigma_1=\mid\Phi^+\rangle\langle\Phi^+\mid$, and $p_2=0$ where $\mid\Phi^+\rangle\langle\Phi^+\mid$ is the maximally entangled state and $\sigma_2$ is an arbitrary state. This gives
\begin{equation}
\begin{aligned}
\frac{\frac{p_{\text {succ }}\left(\left\{p_{i}, \sigma_{i}\right\},\left\{M_{i}\right\}, \mathbb{I}\otimes\Lambda\right)}{p_{\text {succ }}\left(\left\{p_{i}, \sigma_{i}\right\},\left\{N_{i}\right\}, \mathbb{I}\otimes\Lambda\right)}}{\mathop{\max} \limits _{\Xi \in \mathbb{O}} \frac{p_{\text {succ }}\left(\left\{p_{i}, \sigma_{i}\right\},\left\{M_{i}\right\}, \mathbb{I}\otimes\Xi\right)}{p_{\text {succ }}\left(\left\{p_{i}, \sigma_{i}\right\},\left\{N_{i}\right\}, \mathbb{I}\otimes\Xi\right)}}&=\frac{\frac{\langle A, \mathbb{I}\otimes \Lambda(\mid\Phi^+\rangle\langle\Phi^+\mid)\rangle}{\langle B, \mathbb{I}\otimes \Lambda(\mid\Phi^+\rangle\langle\Phi^+\mid)\rangle}}{\mathop{\max} \limits _{\Xi \in \mathbb{O}} \frac{\langle A, \mathbb{I}\otimes \Xi(\mid\Phi^+\rangle\langle\Phi^+\mid)\rangle}{\langle B, \mathbb{I}\otimes \Xi(\mid\Phi^+\rangle\langle\Phi^+\mid)\rangle}}\\
&=\frac{\langle A, \mathcal{J}_\Lambda\rangle}{\langle B, \mathcal{J}_\Lambda\rangle} \min _{\Xi \in \mathbb{O}} \frac{\langle B, \mathcal{J}_{\Xi}\rangle}{\langle A, \mathcal{J}_{\Xi}\rangle}\\
&\geq \frac{\langle A, \mathcal{J}_{\Lambda}\rangle}{\langle B, \mathcal{J}_{\Lambda}\rangle},
\end{aligned}
\end{equation}
where the last inequality follows from any feasible $A$, $B$ satisfying $\langle A, \mathcal{J}_\Xi \rangle\leq \langle B, \mathcal{J}_\Xi \rangle, \forall \Xi\in \mathbb{O}$. Here we have restricted operators $A$, $B$ such that $\langle A, \mathcal{J}_\Lambda \rangle\neq 0\neq \langle B, \mathcal{J}_\Lambda \rangle$ and $\langle A, \mathcal{J}_\Xi \rangle\neq 0\neq\langle B, \mathcal{J}_\Xi \rangle, \forall \Xi\in \mathbb{O}$ to ensure that Eq.(47) is well defined. Since the supremum over all feasible
$A$, $B$ is precisely $\Omega_{\mathbb{O}}(\Lambda)$, this shows that the projective robustness of channels is a lower bound on the LHS of Eq.(47). The proof is completed.
\end{proof}
Note that the above theorem applies to any convex resource theory of channels, thus providing a general operational meaning for the projective robustness of channels.

\section{Projective robustness of measurements}
It is well-known that any meaningful information process involves a measurement at the end, so it is quite natural and necessary to study the quantification and manipulation of quantum measurement resources \cite{Oszmaniec17,Guerini17,Oszmaniec19,Skrzypczyk19, Baek20,Takagi and Regula 19,Oszmaniec19b,Ducuara20,Uola20,Ye et.al 21}. In this section, we will introduce a new measurement resource monotone with a clear operational interpretation. To accomplish this goal, we start by recalling the generalized robustness and weight for measurements \cite{Takagi and Regula 19,Oszmaniec19b,Ducuara20,Uola20}
\begin{equation}
\begin{aligned}
R_{\mathcal{M_F}}(\mathbb{M})&=\left.\inf\{r \ | \ M_i+rN_i=(1+r)F_i, \forall i, \{N_i\}\in \mathcal{M}(d,n), \right.\\
&\qquad \left.\{F_i\}\in \mathcal{M_F}\}\right.,\\
W_{\mathcal{M_F}}(\mathbb{M})&=\left.\inf\{w \ | \ M_i=wN_i+(1-w)F_i , \forall i, \{N_i\}\in \mathcal{M}(d,n),\right.\\
&\qquad \left. \{F_i\}\in \mathcal{M_F}\}\right..
\end{aligned}
\end{equation}
In fact, it is easy to verify that they are equivalent to the following forms:
\begin{equation}
\begin{aligned}
R_{\mathcal{M_F}}(\mathbb{M})=\inf\{\lambda \ | \ M_i\leq \lambda F_i, \forall i, \{F_i\}\in \mathcal{M_F}\},\\
W_{\mathcal{M_F}}(\mathbb{M})=\sup\{\mu \ | \ M_i\geq \mu F_i, \forall i, \{F_i\}\in \mathcal{M_F}\},
\end{aligned}
\end{equation}
where inequality $M_i\leq \lambda N_i$ is understood as ordering of positive semidefinite operators.

\begin{definition}
For a given measurement $\mathbb{M}=\{M_i\}$, its projective robustness with respect to $\mathcal{M_F}$ is defined as
\begin{equation}
\begin{aligned}
\Omega_{\mathcal{M_F}}(\mathbb{M})&=\min_{\mathbb{N}\in \mathcal{M_F}}R_{max}(\mathbb{M}\parallel\mathbb{N})R_{max}(\mathbb{N}\parallel\mathbb{M}),
\end{aligned}
\end{equation}
where $R_{max}(\mathbb{M}\parallel\mathbb{N})=\inf\{\lambda \mid M_i\leq \lambda N_i, \forall i\}$ and $R_{max}(\mathbb{N}\parallel\mathbb{M})=\inf\{\mu \mid N_i\leq \mu M_i, \forall i\}$ for any two measurements $\mathbb{M}=\{M_i\}_i$  and $\mathbb{N}=\{N_i\}_i$.
\end{definition}

\begin{theorem}
The projective robustness of channels $\Omega_{\mathcal{M_F}}(\mathbb{M})$ satisfies the following properties:\\
(i)  $\Omega_{\mathcal{M_F}}(\mathbb{M})$ is finite if and only if there exists a free measurement $\mathbb{N}\in \mathcal{M_F}$ such that $\operatorname{supp}(M_i)=\operatorname{supp}(N_i)$ for all $i$.\\
(ii) $\Omega_{\mathcal{M_F}}(k\mathbb{M})=\Omega_{\mathcal{M_F}}(\mathbb{M})$ for any $k>0$.\\
(iii) When $\mathcal{M_F}$ is a convex set, $\Omega_{\mathcal{M_F}}$ is quasiconvex: for any $t\in[0,1]$, it holds that
\begin{equation}
\Omega_{\mathcal{M_F}}(t\mathbb{M}+(1-t)\mathbb{N})\leq \max \{\Omega_{\mathcal{M_F}}(\mathbb{M}),\Omega_{\mathcal{M_F}}(\mathbb{N})\}.
\end{equation}
(iv) $\Omega_{\mathcal{M_F}}(\Gamma(\mathbb{M}))\leq\Omega_{\mathcal{M_F}}(\mathbb{M})$ for any free operation $\Gamma\in \mathcal{O_F}$, i.e. $\Gamma(\mathcal{M_F})\subseteq \mathcal{M_F}$.\\
(v) When $\mathcal{M_F}$ is a compact convex set, $\Omega_{\mathcal{M_F}}$ can be computed as the optimal value of a conic linear optimization problem:
\begin{equation}
\begin{aligned}
&\Omega_{\mathcal{M_F}}(\mathbb{M})\\
&=\inf\left\{\gamma \in \mathbb{R}_+ \ \Big| \ M_i\leq \widetilde{N}_i\leq \gamma M_i, \{\widetilde{N}_i\}\in \operatorname{cone}(\mathcal{M_F})\right\}\\
&\left.\sup\Bigg\{\sum_i\langle A_i,M_i\rangle \ \Big| \sum_i\langle B_i,M_i\rangle=1,\ A_i,B_i\geq 0 ~\forall i, \right.\\
&~~~~~~~~~~~~~~~~~~~~~~~~~~~~~~~~~~~~~~~~~~~~~~~~~~~~~~~\left.\{B_i-A_i\}\in \operatorname{cone}(\mathcal{M_F})^* \Big\}\right.\\
&=\sup\left\{\frac{\sum_i\langle A_i,M_i\rangle}{\sum_i\langle B_i,M_i\rangle} \ \Big| \  A_i,B_i\geq 0 ~\forall i, \frac{\sum_i\langle A_i,N_i\rangle}{\sum_i\langle B_i,N_i\rangle}\leq 1, \forall \{N_i\}\in \mathcal{M_F} \right\},
\end{aligned}
\end{equation}
where $\operatorname{cone}(\mathcal{M_F})=\{\lambda \mathbb{N} \mid \lambda \in \mathbb{R}_+, \mathbb{N}\in \mathcal{M_F}\}$ is the closed convex cone generated by the set $\mathcal{M_F}$, and $\operatorname{cone}(\mathcal{M_F})^*=\{\mathbb{X} \mid \langle \mathbb{X},\mathbb{N}\rangle\geq 0, \forall \mathbb{N} \in \mathcal{M_F}\}$. Note that the compactness of $\mathcal{M_F}$ leads to the fact that the infimum is achieved as long as it is finite.\\
(vi) It can be bounded as
\begin{equation}
\begin{aligned}
R_{\mathcal{M_F}}(\mathbb{M})W_{\mathcal{M_F}}(\mathbb{M})^{-1}\leq&\Omega_{\mathcal{M_F}}(\mathbb{M})\leq R_{\mathcal{M_F}}(\mathbb{M})R_{\max}(\mathbb{N}^*_R\parallel \mathbb{M})\\
&\Omega_{\mathcal{M_F}}(\mathbb{M})\leq W_{\mathcal{M_F}}(\mathbb{M})^{-1}R_{\max}(\mathbb{M}\parallel \mathbb{N}^*_W),
\end{aligned}
\end{equation}
where $\mathbb{N}_R^*$ is an optimal channel such that $R_{\mathcal{M_F}}(\mathbb{M})=R_{\max}(\mathbb{M}\parallel\mathbb{N}^*_R)$, and similarly $\mathcal{M}_W^*$ is an optimal channel such that $W_{\mathcal{M_F}}(\mathbb{M})=R_{\max}(\mathbb{N}^*_W\parallel\mathbb{M})^{-1}$, whenever such channels exist.\\
\end{theorem}
\begin{proof}
The proof is similar as Theorem 1, and the proof is given in ``Appendix A" for completeness.
\end{proof}

{\textbf{Interpretation in state discrimination.}} As a application of the projective robustness $\Omega_{\mathcal{M_F}}$, we will show that $\Omega_{\mathcal{M_F}}(\mathbb{M})$ quantifies the advantage that any resourceful measurement provides over all resourceless measurements when state discrimination and exclusion tasks are considered at the same time.

\begin{theorem}
The projective robustness of measurements quantifies the maximal advantage that a given measurement $\mathbb{M}=\{M_i\}$ gives over all free measurements $\mathbb{N}\in\mathcal{M_F}$ in simultaneous discrimination of a fixed state ensemble $\{p_i,\sigma_i\}$ and exclusion of a fixed state ensemble $\{q_i,\tau_i\}$, as quantified by the ratio
\begin{equation}
\frac{p_{\text{succ}}\left(\left\{p_{i}, \sigma_{i}\right\},\left\{M_{i}\right\}\right)}{p_{\text{succ}}\left(\left\{q_{i}, \tau_{i}\right\},\left\{M_{i}\right\}\right)},
\end{equation}
where we aim to perform state discrimination and exclusion with the same measurement $\{M_i\}$.

Specifically,
\begin{equation}
\begin{aligned}
\sup _{\left\{p_{i}, \sigma_{i}\right\} \atop\left\{q_{i},\tau_i\right\}} \frac{\frac{p_{\text {succ }}\left(\left\{p_{i}, \sigma_{i}\right\},\left\{M_{i}\right\}\right)}{p_{\text {succ }}\left(\left\{q_{i}, \tau_{i}\right\},\left\{M_{i}\right\}, \right)}}{\mathop{\max} \limits _{\mathbb{N} \in \mathcal{M_F}} \frac{p_{\text {succ }}\left(\left\{p_{i}, \sigma_{i}\right\},\left\{N_{i}\right\}\right)}{p_{\text {succ }}\left(\left\{q_{i}, \tau_{i}\right\},\left\{N_{i}\right\}\right)}}=\Omega_{\mathcal{M_F}}(\mathbb{M}),
\end{aligned}
\end{equation}
where the maximization is over all finite state ensembles $\{p_i,\sigma_i\}_{i=1}^n$ and $\{q_i,\tau_i\}_{i=1}^n$ for which the expression in Eq.(57) is well defined.
\end{theorem}
\begin{proof}
Let us assume that $\Omega_{\mathcal{M_F}}(\mathbb{M})$ is finite and take $M_i\leq \lambda N_i$ and $N_i\leq \mu M_i$ such that $\Omega_{\mathcal{M_F}}(\mathbb{M})=\lambda\mu$. Let $\{p_i,\sigma_i\}$ and $\{q_i,\tau_i\}$ be any state ensembles, and $\{M_i\}_{i=1}^n$ be a given POVM measurement. Then
\begin{equation}
\begin{aligned}
\frac{p_{\text {succ }}\left(\left\{p_{i}, \sigma_{i}\right\},\left\{M_{i}\right\}\right)}{p_{\text {succ }}\left(\left\{q_{i}, \tau_{i}\right\},\left\{M_{i}\right\}\right)}& \leq \frac{\lambda p_{\text {succ }}\left(\left\{p_{i}, \sigma_{i}\right\},\left\{N_{i}\right\}\right)}{\mu^{-1} p_{\text {succ }}\left(\left\{q_{i}, \tau_{i}\right\},\left\{N_{i}\right\} \right)} \\
& \leq \lambda \mu \max _{\mathbb{N} \in \mathcal{M_F}} \frac{p_{\text {succ }}\left(\left\{p_{i}, \sigma_{i}\right\},\left\{N_{i}\right\}\right)}{ p_{\text {succ }}\left(\left\{q_{i}, \tau_{i}\right\},\left\{N_{i}\right\},\right)},
\end{aligned}
\end{equation}
where the first inequality follows by the linearity of $p_{\text{succ}}$. This implies that $\Omega_{\mathcal{M_F}}(\mathbb{M})$ always upper bounds the LHS of Eq.(57).

Consider any dual feasible solutions $\{A_i\}$ and $\{B_i\}$ for $\Omega_{\mathcal{M_F}}(\mathbb{M})$ in dual problem Eq.(54). We define the state ensembles $\{p_i,\sigma_i\}$ and $\{q_i,\tau_i\}$ as
\begin{equation}
\{p_i,\sigma_i\}=\left\{\begin{array}{ll}\vspace{1ex}
{\left\{\frac{\langle \mathbb{I},A_i\rangle}{\sum_i\langle \mathbb{I},A_i\rangle}, \frac{A_i}{\langle \mathbb{I},A_i\rangle}\right\},} &
{\text { if }\langle \mathbb{I},A_i\rangle> 0}, \\ \vspace{1ex}
{\{0, \sigma\},} & {\text { if } \langle \mathbb{I},A_i\rangle=0},
\end{array}\right.
\end{equation}
where $\sigma$ is an arbitrary state, and
\begin{equation}
\{q_i,\tau_i\}=\left\{\begin{array}{ll}\vspace{1ex}
{\left\{\frac{\langle \mathbb{I},B_i\rangle}{\sum_i\langle \mathbb{I},B_i\rangle}, \frac{B_i}{\langle \mathbb{I},B_i\rangle}\right\},} &
{\text { if }\langle \mathbb{I},B_i\rangle> 0}, \\ \vspace{1ex}
{\{0, \sigma\},} & {\text { if } \langle \mathbb{I},B_i\rangle=0},
\end{array}\right.
\end{equation}
where $\tau$ is an arbitrary state. For both state ensembles, we get
\begin{equation}
\begin{aligned}
\frac{\frac{p_{\text {succ }}\left(\left\{p_{i}, \sigma_{i}\right\},\left\{M_{i}\right\}\right)}{p_{\text {succ }}\left(\left\{q_{i}, \tau_{i}\right\},\left\{M_{i}\right\}\right)}}{\mathop{\max} \limits _{\mathbb{N} \in \mathcal{M_F}} \frac{p_{\text {succ }}\left(\left\{p_{i}, \sigma_{i}\right\},\left\{N_{i}\right\}\right)}{p_{\text {succ }}\left(\left\{q_{i}, \tau_{i}\right\},\left\{N_{i}\right\}\right)}}&=\frac{\frac{\sum_i\langle M_i,A_i\rangle}{\sum_i\langle M_i,B_i\rangle}}
{\mathop{\max} \limits _{\mathbb{N} \in \mathcal{M_F}} \frac{\sum_i\langle N_i,A_i\rangle}{\sum_i\langle N_i,B_i\rangle}}\\
&\geq \frac{\sum_i\langle M_i,A_i\rangle}{\sum_i\langle M_i,B_i\rangle},
\end{aligned}
\end{equation}
where the last inequality is since any feasible $A_i$, $B_i$ satisfy $\sum_i\langle N_i,A_i\rangle\leq \sum_i\langle N_i,B_i\rangle, \forall \{N_i\}\in \mathcal{M_F}$. Here we have constrained ourselves to the operators $A_i$, $B_i$ such that $\sum_i\langle M_i, A_i\rangle\neq 0\neq \sum_i\langle M_i,B_i\rangle$ and $\sum_i\langle N_i,A_i\rangle\neq 0\neq\sum_i\langle N_i,B_i\rangle, \forall \{N_i\}\in \mathcal{M_F}$ to ensure that Eq.(57) is well defined. Since the supremum over all feasible
$A_i$, $B_i$ is precisely $\Omega_{\mathcal{M_F}}(\mathbb{M})$, this establishes the projective robustness of channels as a lower bound on the LHS of Eq.(57). The proof is completed.
\end{proof}
Note that the above theorem applies to any convex resource theory of measurements, thus providing a general operational meaning for the projective robustness of measurements.

In Ref.\cite{Skrzypczyk19b}, it is shown that any incompatible measurement gives an advantage over all compatible measurements in quantum state discrimination. It is easy to see that the set of all compatible measurements forms a convex resource theory of measurements. Next, we will discuss the projective robustness of measurement incompatibility.

{\textbf{Measurement incompatibility.}} The set of POVMs $\left\{M_{a \mid x}\right\}$, where $M_{a \mid x}$ is the POVM element with
outcome $a$ for the measurement labeled by setting $x$, is called compatible (or jointly measurable) \cite{Guerini17} if there exists a parent measurement $\left\{N_{i}\right\}$ and a conditional probability distribution $\left\{q(a \mid x, i)\right\}$ such that $M_{a \mid x}=\sum_{i} q(a \mid x, i) N_{i}$, where $a=1,...,n$ and $x=1,...,m$.  Let $\mathfrak{M}$ denote all sets of POVMs, and $\mathfrak{C}$ denote the set of compatible POVMs. For any set of POVMs $\left\{M_{a \mid x}\right\}$, we can directly define the projective robustness of measurement incompatibility as
\begin{equation}
\begin{aligned}
\Omega_{\mathfrak{C}}\left(\left\{M_{a \mid x}\right\}\right)&=\left.\inf\big\{\lambda\mu \mid M_{a \mid x}\leq \lambda N_{a \mid x}, N_{a \mid x}\leq \mu M_{a \mid x}, \forall a, x,\right.\\
&\qquad \left.\{N_{a \mid x}\} \in \mathfrak{C}\big\}\right..
\end{aligned}
\end{equation}

In Ref.\cite{Regula and Takagi21b}, a specific theory of channel based on the theory of measurement incompatibility was formulated by taking $\mathbb{O}_{\text {all }}$ to be the set of channels representing the sets of POVMs as
\begin{equation}
\mathbb{O}_{\text {all }}:=\left\{\mathcal{E}_{\left\{M_{a \mid x}\right\}} \mid\left\{M_{a \mid x}\right\} \in \mathfrak{M}\right\}
\end{equation}
where we defined
\begin{equation}
\mathcal{E}_{\left\{M_{a \mid x}\right\}}(\sigma \otimes \rho):=\sum_{x, a}\left\langle x|\sigma| x\right\rangle\left\langle M_{a \mid x} ,\rho\right\rangle|a\rangle\langle a|
\end{equation}
where $\{|x\rangle\}$ and $\{|a\rangle\}$ are orthonormal bases representing classical variables for measurement settings and measurement outcomes, respectively. Naturally, the set of free channels $\mathbb{O} \subseteq \mathbb{O}_{\text {all }}$ was defined using the set of compatible POVMs $\mathfrak{C}$ as
\begin{equation}
\mathbb{O}:=\left\{\mathcal{E}_{\left\{M_{a \mid x}\right\}} \mid\left\{M_{a \mid x}\right\} \in \mathfrak{C}\right\}.
\end{equation}

\begin{proposition} Let $\mathbb{O}_{\text {all }}$ be the set of bipartite channels defined in (63) and $\mathbb{O}$ be the set of free channels that represents the compatible POVMs. Then, it holds that
\begin{equation}
\Omega_{\mathfrak{C}}\left(\left\{M_{a \mid x}\right\}\right)=\Omega_{\mathbb{O}}(\mathcal{E}_{\{M_{a \mid x}\}}).
\end{equation}
\end{proposition}
\begin{proof}
Following the definition of $\Omega_{\mathbb{O}}$,
\begin{equation}
\begin{split}
&\Omega_{\mathbb{O}}\left(\mathcal{E}_{\left\{M_{a \mid x}\right\}}\right)=\inf \left\{\lambda\mu \mid \mathcal{E}_{\left\{M_{a \mid x}\right\}}\leq \gamma\mathcal{F},\mathcal{F}\leq \mu\mathcal{E}_{\left\{M_{a \mid x}\right\}}, \mathcal{F} \in \mathbb{O}\right\} \\
&\quad=\inf \left\{\lambda\mu \mid \mathcal{E}_{\left\{M_{a \mid x}\right\}}\leq \lambda\mathcal{E}_{\left\{F_{a \mid x}\right\}}, \mathcal{E}_{\left\{F_{a \mid x}\right\}}\leq \mu\mathcal{E}_{\left\{M_{a \mid x}\right\}}, \left\{F_{a \mid x}\right\} \in \mathfrak{C}\right\} \\
&\quad=\inf \left\{\lambda\mu \ | \ \forall \sigma, \rho, \sum_{x, a}\langle x|\sigma| x\rangle  \left\langle\lambda F_{a \mid x}-M_{a \mid x}, \rho\right\rangle|a\rangle\langle a|\geq 0,\right. \\
&\qquad\left.\quad\sum_{x, a}\langle x|\sigma| x\rangle  \left\langle\mu M_{a \mid x}- F_{a \mid x}, \rho\right\rangle|a\rangle\langle a|\geq 0,  \left\{F_{a \mid x}\right\} \in \mathfrak{C}\right\} \\
&\quad=\inf \left\{\lambda\mu \mid M_{a \mid x}\leq \lambda F_{a \mid x}, F_{a \mid x}\leq \mu M_{a \mid x} \forall a, x,\left\{F_{a \mid x}\right\} \in \mathfrak{C}\right\},
\end{split}
\end{equation}
where the last equality is obtained by observing that in the expression of the the third line, the fact that two inequalities hold for any $\sigma$, as well as that $|a\rangle\langle a|$ is a classical state, implies that $\left\langle\lambda F_{a \mid x}-M_{a \mid x}, \rho\right\rangle\geq 0$ and $\left\langle\mu M_{a \mid x}- F_{a \mid x}, \rho\right\rangle\geq 0$, $\forall a,x$,
and further imposing that this holds for any $\rho$ gives $\lambda F_{a \mid x}-M_{a \mid x}\geq0$ and $\mu M_{a \mid x}- F_{a \mid x}\geq0$, $\forall a,x$. On the other hand, $\lambda F_{a \mid x}-M_{a \mid x}\geq0$ and $\mu M_{a \mid x}- F_{a \mid x}\geq0$, $\forall a,x$ imply the expression in the third line. Note that the last expression is precisely $\Omega_{\mathfrak{C}}\left(\left\{M_{a \mid x}\right\}\right)$. The proof is completed.
\end{proof}
Proposition 2 shows that the projective robustness of measurement incompatibility coincides with the channel-based measure defined in our framework. This helps to further study resource manipulation of incompatible measurements.

\section{Conclusion}

In this work, we introduced projective robustness for channels and measurements in any convex quantum resource theory and discussed their operational interpretations in simultaneous discrimination and exclusion of states. First, we defined the projective robustness of channels in any convex dynamic resource theory and showed that it satisfies some good properties such as quasiconvexity, invariance under scaling, lower semicontinuity, and sub- or supermultiplicativity. Importantly, it can always be computed as a convex (conic) optimization problem. Moreover, the projective robustness of channels provided the lower bounds on the error and many-copy distillation overhead in any deterministic channel resource distillation task.  Meanwhile, the projective robustness of channels could be regarded as the maximal advantage provided by a given resourceful channel over all resourceless channels in simultaneous discrimination and exclusion of a fixed state ensemble by an application of the bipartite channel. Second, we defined the projective robustness of measurements in any convex resource theory of measurements. Similarly, we proved that it has some good properties and showed that it can exactly quantify the maximal advantage provided by a given resourceful measurement over all resourceless measurements in the simultaneous discrimination and exclusion of two fixed state ensembles. Finally, in a specific channel resource setting based on measurement incompatibility, we found an equivalence relation between the projective robustness for quantum channels and the projective robustness for measurement incompatibility.

Our results apply to general convex resource theories of channels and measurements. We hope that the work can deepen one's comprehension of the quantification and manipulation of channel resources and effectively facilitate one's capacity of exploiting dynamic resources.

\section{Acknowledge}
This paper was supported by National Science Foundation of China (Grant Nos:12071271, 11671244, 62001274), the Higher School Doctoral Subject Foundation of Ministry of Education of China (Grant No:20130202110001) and the Research Funds for the Central Universities (GK202003070).

\section*{Appendix A: The proof of Theorem 3}
The proofs of (i)-(iii) are similar as the (i)-(iii) of  Theorem 3.

(iv) If $\Omega_{\mathbb{O}}(\mathbb{M})=\infty$, then the result is trivial, so we shall assume
otherwise. Let $\{F_i\}\in \mathcal{M_F}$ be a free measurement such that $M_i\leq\lambda F_i$ and $F_i\leq\mu M_i$ for any $i$ with $\Omega_{\mathcal{M_F}}(\mathbb{M})=\lambda\mu$. Since $\Gamma$ is a free operation, we have $\{\Gamma(F_i)\}_i\in \mathcal{M_F}$ , and $\Gamma(M_i)\leq\lambda \Gamma(F_i)$ and $\Gamma(F_i)\leq\mu \Gamma(M_i)$ for any $i$. So we can directly get $\Omega_{\mathcal{M_F}}(\Gamma(\mathbb{M}))\leq \Omega_{\mathcal{M_F}}(\mathbb{M})$.

(v) We can see that
\begin{equation}
\begin{aligned}
\Omega_{\mathcal{M_F}}(\mathbb{M})&=\min_{\mathbb{N}\in \mathcal{M_F}}\left[ \inf\{\lambda \mid M_i\leq \lambda N_i, \forall i\}
\inf\{\mu \mid N_i\leq \mu M_i, \forall i\}\right]\\
&=\inf\{\lambda\mu \mid M_i\leq \lambda N_i, N_i\leq \mu M_i, \forall i, \{N_i\}\in \mathcal{M_F}\}.
\end{aligned}
\end{equation}
Let $\widetilde{N}_i=\lambda N_i$ for all $i$. Observe that any feasible solution to the problem
\begin{equation}
\inf\left\{\gamma \ \Big| \ M_i\leq \widetilde{N}_i\leq \gamma M_i, \{\widetilde{N}_i\}\in \operatorname{cone}(\mathcal{M_F})\right\}
\end{equation}
gives a feasible solution to Eq.(68) as
\begin{equation}
N_i=\frac{\widetilde{N}_i}{\frac{\sum_i\langle \mathbb{I},\widetilde{N}_i\rangle}{d}}, \lambda=\frac{\sum_i\langle \mathbb{I},\widetilde{N}_i\rangle}{d}, \mu=\frac{\gamma}{\frac{\sum_i\langle \mathbb{I},\widetilde{N}_i\rangle}{d}}
\end{equation}
with objective function value $\lambda\mu=\gamma$. Conversely, any feasible solution $\{\{N_i\}, \lambda, \mu\}$ to Eq.(69) gives a feasible solution to Eq.(68) as $\widetilde{N}_i=\lambda N_i, \forall i$ and $ \gamma=\lambda\mu$. Thus, the two problems are equivalent.\\
Writing the Lagrangian as
\begin{equation}
\begin{split}
&L(\gamma,\widetilde{N}_i;\{A_i\},\{B_i\},\{C_i\})\\
&=\gamma-\sum_i\langle A_i,\widetilde{N}_i-M_i\rangle-\sum_i\langle B_i,\gamma M_i-\widetilde{N}_i\rangle-\sum_i\langle C_i,\widetilde{N}_i\rangle\\
&=\gamma[1-\sum_i\langle B_i,M_i\rangle]+\sum_i\langle A_i,M_i\rangle+\sum_i\langle B_i-A_i-C_i,\widetilde{N}_i\rangle.
\end{split}
\end{equation}
Optimizing over the Lagrange multipliers $A_i, B_i\geq 0$ for all $i$, and $\{C_i\}\in \operatorname{cone}(\mathcal{M_F})^*$ , the corresponding dual form of primal problem Eq.(69) is written as
\begin{equation}
\begin{aligned}
&\left.\sup\Bigg\{\sum_i\langle A_i,M_i\rangle \ \Big| \sum_i\langle B_i,M_i\rangle=1,\ A_i,B_i\geq 0 ~\forall i, \right.\\
&~~~~~~~~~~~~~~~~~~~~~~~~~~~~~~~~~~~~~~~~~~~~~~~~~~~~~~~\left.\{B_i-A_i\}\in \operatorname{cone}(\mathcal{M_F})^* \Bigg\}\right.\\
&=\sup\left\{\frac{\sum_i\langle A_i,M_i\rangle}{\sum_i\langle B_i,M_i\rangle} \ \Big| \  A_i,B_i\geq 0 ~\forall i, \frac{\sum_i\langle A_i,N_i\rangle}{\sum_i\langle B_i,N_i\rangle}\leq 1, \forall \{N_i\}\in \mathcal{M_F} \right\}.
\end{aligned}
\end{equation}
If we take $B_i=\mathbb{I}$ and $A_i=\epsilon\mathbb{I}$ for all $i$ where $0<\epsilon<1$, it is obvious that $A_i$ and $B_i$ are strictly feasible for the dual. Thus, it follows from Slater's theorem \cite{Boyd04} that strong duality holds, the optimal value of the primal problem Eq.(69) is equal to that of dual problem Eq.(72).

The second line of Eq.(72) follows since any feasible solution to this program can be rescaled as $A_i\mapsto\frac{A_i}{\sum_i \langle B_i,M_i\rangle}$, $B_i\mapsto\frac{B_i}{\sum_i \langle B_i,M_i\rangle}$ to give a feasible solution to the dual, and vice versa. Here, we implicitly constrain ourselves to $B_i$ such that $\sum_i \langle B_i,M_i\rangle\neq 0$ and $\sum_i\langle B_i,N_i\rangle\neq 0 ~\forall \{N_i\}\in\mathcal{M_F}$; this can always be ensured by taking $B_i$ as a small multiple of the identity.

(vi) The lower bound is obtained by noting that
\begin{equation}
\begin{aligned}
\Omega_{\mathcal{M_F}}(\mathbb{M}) &=\min_{\mathbb{N} \in \mathcal{M_F}} R_{\max }(\mathbb{M} \| \mathbb{N}) R_{\max }(\mathbb{N} \| \mathbb{M}) \\
& \geq\left[\min _{\mathbb{N} \in \mathcal{M_F}} R_{\max }(\mathbb{M} \| \mathbb{N})\right]\left[\min _{\mathbb{N} \in \mathcal{M_F}} R_{\max }(\mathbb{N} \| \mathbb{M})\right] \\
&=R_{\mathcal{M_F}}(\mathbb{M}) W_{\mathcal{M_F}}(\mathbb{M})^{-1}.
\end{aligned}
\end{equation}
The upper bounds follow by using $\mathbb{N}^*_{R}$ and $\mathbb{N}^*_{W}$ as feasible solutions in the definition of $\Omega_{\mathcal{M_F}}$.

\end{document}